\documentclass[smallextended]{svjour3}       
\smartqed  
\usepackage[T1]{fontenc}
\usepackage[utf8]{inputenc}
\usepackage{subfig,multirow}
\usepackage{dcolumn}
\usepackage[]{hyperref}
\usepackage{amsmath,amsfonts,amssymb,graphicx,url,latexsym,tabularx,xspace,tikz}
\usepackage[noend, vlined, algoruled, linesnumbered] {algorithm2e}

\newcommand{\softO}{\ensuremath{\widetilde{O}}\xspace}
\def\polylog{\operatorname{polylog}}
\newcommand{\grafopesato}{\ensuremath{G=(V,E,w)}\xspace}
\newcommand{\firstfpt}{\textsc{\clubfs--Fpt1}\xspace}
\newcommand{\secondfpt}{\textsc{\clubfs--Fpt2}\xspace}
\newcommand{\ptas}{\ensuremath{\mbox{\textsf{PTAS}}}\xspace}
\newcommand{\np}{\ensuremath{\mbox{\textsf{NP}}}\xspace}
\newcommand{\p}{\ensuremath{\mbox{\textsf{P}}}\xspace}
\newcommand{\cluspt}{\textsc{CluSPT}\xspace}
\newcommand{\clubfs}{\textsc{CluBFS}\xspace}
\newcommand{\clusp}{\textsc{CluSP}\xspace}

\newcommand{\tsat}{\ensuremath{3\mbox{\textsf{--CNF--SAT}}}\xspace}
\newcommand{\xc}{\ensuremath{\mbox{\textsf{Exact--Cover--by--3-Sets}}}\xspace}
\newcommand{\sat}{\ensuremath{\mbox{\textsf{CNF--SAT}}}\xspace}
\newcommand{\V}{\ensuremath{\mathcal{V}}\xspace}
\newcommand{\varset}{\ensuremath{H}\xspace}
\newcommand{\varsetprime}{\ensuremath{H'}\xspace}

\newcommand{\cluster}[1]{\ensuremath{V_{#1}}}
\newcommand{\diam}{\textsc{diam}}
\newcommand{\cost}{\textsc{cost}}

\newcommand{\cluvertex}[1]{\ensuremath{\nu_{#1}}}
\newcommand{\cluvertexopt}[1]{\ensuremath{\tau_{#1}}}
\newcommand{\approxsol}{\ensuremath{\tilde{T}}\xspace}
\newcommand{\clustertree}{\ensuremath{{T_C}}\xspace}

\newcommand{\optclustertree}{\ensuremath{{T_C}^*}\xspace}
\newcommand{\optsol}{\ensuremath{{T}^*}\xspace}
\newcommand{\optsolident}{\ensuremath{T''}\xspace}
\newcommand{\clustbfstree}{\ensuremath{T'}\xspace}

\let\tilde\widetilde

\newcommand{\opt}{\ensuremath{\mathtt{OPT}}\xspace}
\newcommand{\bfs}[1]{\ensuremath{\mathtt{BFS}_{#1}}\xspace}

\newcommand{\phigraph}{\ensuremath{G_{\phi}}\xspace}
\newcommand{\vphi}{\ensuremath{V(\phigraph)}\xspace}
\newcommand{\ephi}{\ensuremath{E(\phigraph)}\xspace}

\newcommand{\param}{\ensuremath{h}\xspace}

\newcommand{\I}{\mathcal{I}}
\renewcommand{\S}{\mathcal{S}}

\newcommand{\totalfpt}{\ensuremath{\softO\big(\min\big\{2^k k^3 n^4, \; \param^\frac{\param}{2} m\big\}\big)}\xspace}
\newcommand{\totalfptspt}{\ensuremath{\softO\big(\min\big\{nm + 2^k k^3 n^2 \cdot \opt \, \log \opt, \; \param^\frac{\param}{2} m\big\}\big)}\xspace}

\begin{document}

\title{Hardness, Approximability, and Fixed-Parameter Tractability of the Clustered Shortest-Path Tree Problem
\thanks{The results presented in this work have been announced in a preliminary form in \cite{DFFLP16}.}
}

\titlerunning{Clustered Shortest-Path Tree Problem}        

\author{Mattia~D'Emidio  \and
        Luca~Forlizzi  \and
        Daniele~Frigioni \and
        Stefano~Leucci \and
        Guido~Proietti
}


\institute{M. D'Emidio \at
              Gran Sasso Science Institute, L'Aquila, Italy. \\
              \email{mattia.demidio@gssi.it} .
           \and
           L. Forlizzi, D. Frigioni, G. Proietti \at
           Department of Information Engineering, Computer Science and Mathematics,\\
           University of L'Aquila, L'Aquila, Italy.\\
           \email{luca.forlizzi@univaq.it, daniele.frigioni@univaq.it, guido.proietti@univaq.it}.
           \and
           S. Leucci \at
           Department of Computer Science,  ETH Z\"urich, Switzerland. \\ \email{stefano.leucci@inf.ethz.ch}.
           \and
           G. Proietti \at
           Istituto di Analisi dei Sistemi e Informatica ``Antonio Ruberti''  Consiglio Nazionale delle Ricerche, Roma, Italy.
}

\date{Received: date / Accepted: date}

\maketitle
\begin{abstract}
Given an $n$-vertex non-negatively real-weighted graph $G$, whose vertices are
partitioned into a set of $k$ clusters, a \emph{clustered network design problem} on $G$ consists of solving a given network design optimization problem
on $G$, subject to some additional constraint on its clusters.
In particular, we focus on the classic problem of designing a \emph{single-source shortest-path tree}, and we analyze its computational hardness when in a feasible solution each cluster is required to form a subtree. We first study the \emph{unweighted} case, and prove that the problem is \np-hard. However, on the positive side, we show the existence of an approximation algorithm whose quality essentially depends on few parameters, but which remarkably is an $O(1)$-approximation when the largest out of all the \emph{diameters} of the clusters is either $O(1)$ or $\Theta(n)$. Furthermore, we also show that the problem is \emph{fixed-parameter tractable} with respect to $k$ or to the number of vertices that belong to clusters of size at least 2. Then, we focus on the \emph{weighted} case, and show that the problem can be approximated within a tight factor of $O(n)$, and that it is fixed-parameter tractable as well. Finally, we analyze the unweighted \emph{single-pair shortest path problem}, and we show it is hard to approximate within a (tight) factor of $n^{1-\epsilon}$, for any $\epsilon>0$.

\keywords{Clustered Shortest-Path Tree Problem \and Hardness \and Approximation Algorithms \and Fixed-Parameter Tractablity \and Network Design}
\end{abstract}

\section{Introduction} \label{sec:intro}
In several modern network applications, the underlying set of nodes may be partitioned
into \emph{clusters}, with the intent of modeling some aggregation phenomena
taking place among similar entities in the network. In particular, this happens in communication and social networks, where clusters may refer
to local-area subnetworks and to communities of individuals, respectively. While
on one hand the provision of clusters allows to represent the complexity of
reality, on the other hand it may ask for introducing some additional
constraints on a feasible solution to a given network design problem, with the
goal of preserving a specific cluster-based property. Thus, on a theoretical
side, given a (possibly weighted) graph $G$,
whose vertex set is partitioned into $k$ pairwise disjoint subset (i.e., clusters),
a \emph{clustered} (a.k.a. \emph{generalized}) \emph{network design problem} on $G$ consists of finding a (possibly
optimal) solution to a given network design problem on $G$, subject to some additional constraint on its clusters.
Depending on such constraint, the computational
complexity of the resulting problem may change drastically as compared to the
unconstrained version. Therefore, this class of problems deserves a theoretical
investigation that, quite surprisingly, seems to be rather missing up to now.

One of the most intuitive constraints one could imagine is that of maintaining
some sort of \emph{proximity} relationship among nodes in a same cluster. This
scenario has immediate practical motivations: for instance, in a communication
network, this can be convincingly justified with the requirement of designing a
network on a classic two-layer (i.e., local \emph{versus} global layer)
topology. In particular, if the foreseen solution has to be a (spanning)
tree $T$ in $G$, then a natural requirement is that each cluster should
induce a (connected) subtree of $T$. For the sake of simplicity, in the
following this will be referred to as a \emph{clustered tree design problem}
(CTDP), even if this is a slight abuse of terminology.
Correspondingly, classic spanning-tree optimization problems on graphs can be reconsidered under
this new perspective, aiming at verifying whether they exhibit a significant
deviation (from a computational point of view) w.r.t. the ordinary (i.e.,
non-clustered) counterpart.
In particular, we will focus on the clustered
version of the problem of computing a \emph{single-source
shortest-path tree} (SPT) of $G$, i.e., a spanning tree of $G$ rooted at a given
source node, say $s$, minimizing the total length of all the paths emanating
from $s$. It is worth noticing that an SPT, besides its theoretical relevance,
has countless applications, and in particular it supports a set of primitives of
primary importance in communication networks, as for instance the
\emph{broadcast protocol} and the \emph{spanning tree protocol}.

\subsection{Contribution of the paper}
Let \grafopesato be a connected and undirected graph of $n$ vertices and $m$ edges, where each edge $(u,v) \in E$ is associated with a non-negative real weight $w(u,v)$. For a subgraph $H$ of $G$, we will use $V(H)$ ($E(H)$, resp.)  to denote the set of vertices (edges, resp.) of $H$, and $H[S]$  to denote the subgraph of $H$ induced by $S$, $S\subseteq V(H)$. Moreover, $\pi_H(u,v)$ will denote a \emph{shortest path} between vertices $u$ and $v$ in $H$, while $d_H(u,v)$ will denote the corresponding \textit{distance} between $u$ and $v$ in $H$, i.e., the sum of the weights of the edges in $\pi_H(u,v)$.

Formally, the clustered version of the SPT problem (\cluspt in the sequel), is defined as follows. We are given a graph $G$ defined as above, along with a partition of $V$ into a set of $k$ (pairwise disjoint) clusters $\V = \{\cluster{1},\cluster{2}, \ldots, \cluster{k}\}$, and a distinguished source vertex $s\in V$. The objective is to find a \textit{clustered SPT} of $G$ rooted at $s$, i.e., a spanning tree $T$ of $G$ such that $T[\cluster{i}],~i=1, \ldots, k$, is a connected component (i.e., a subtree) of $T$, and for which the \emph{broadcast cost} from $s$, i.e. $\cost(T) = \sum_{v \in V} d_{T}(s,v)$, is minimum.

The SPT problem in a non-clustered setting has been widely studied, and in its
more general definition it can be solved in $O(m + n \log n)$ time by means of
the classic Dijkstra's algorithm. More efficient solutions are known for
special classes of graphs (e.g., euclidean, planar, directed acyclic graphs,
etc.), or for restricted edge weights instances. In particular, if $w$ is
uniform, namely $G$ is \textit{unweighted}, then an optimal solution can be
found in $O(m+n)$ time by means of a simple \emph{breadth-first search} (BFS)
visit of $G$. Nevertheless, to the best of our knowledge nothing is known about
its clustered variant, despite the fact that, as we argued above, it is very
reasonable to imagine a scenario where an efficient broadcast needs to be
applied locally and hierarchically within each cluster.
%

Here, we then try to fill this gap, and we show that \cluspt, and its unweighted
version, say \clubfs, are actually much harder than their standard counterparts,
namely:

\begin{enumerate}
\item \clubfs is \np-hard, but  it admits an $O(\min\{\frac{4n k}{\gamma},\frac{4n^2}{\gamma^2},2\gamma\})$ approximation algorithm, where $\gamma$ denotes the length of the largest out of all the \emph{diameters} of the clusters. Interestingly, the approximation ratio becomes $O(1)$ when $\gamma$ is either $O(1)$ or $\Theta(n)$, which may cover cases of practical interest. However, we also point out that in the worst case, namely for $\gamma=\Theta(n^{\frac{2}{3}})$ and $k=\Theta(\sqrt[3]{n})$, the algorithm becomes $O(n^{\frac{2}{3}})$-approximating. Besides that, we also show that the problem is \emph{fixed-parameter tractable}, as we can provide a \totalfpt time exact algorithm,\footnote{Throughout the paper, the notation \softO suppresses
factors that are polylogarithmic in $n$.} where \param is the total number of vertices that belong to clusters of size at least two.

\item \cluspt is hard to approximate within a factor of $n^{1-\epsilon}$ for any constant $\epsilon \in (0, 1]$, unless $\p=\np$, but, on the positive side: (i) it admits an $n$-approximation, thus essentially tight, algorithm; (ii) similarly to the unweighted case, it is fixed-parameter tractable as well.
\end{enumerate}
Finally, we study the \emph{clustered single-pair shortest path problem} (say \clusp in the sequel)
on unweighted graphs, i.e., the problem of finding a shortest path between a
given pair of vertices of $G$, subject to the constraint that the vertices from
a same cluster that belong to the path must appear consecutively. Notice that in this variant, not \emph{all} the vertices of a cluster must belong to a solution, and not \emph{all} the clusters must enter
into a solution. We show that it cannot be approximated in polynomial time
within a factor of $n^{1-\epsilon}$, for any constant $\epsilon > 0$, unless
$\p=\np$. This extends the inapproximability result (within any polynomial factor) that was given in \cite{LW16} for the corresponding weighted version. Since obtaining
an $n$-approximation is trivial, the provided inapproximability result is a bit
surprising, as one could have expected the existence of a $o(n)$-approximation
algorithm, similarly to what happened for \clubfs.

\subsection{Related Work}
Several classic tree/path-design problems have been investigated in the CTDP framework. Some of them, due to their nature, do not actually exhibit a significant deviation (from a computational point of view) w.r.t. the ordinary (i.e., non-clustered) counterpart. For instance, the \emph{minimum spanning tree} (MST) problem falls in this category, since we can easily solve its clustered version by first computing a MST of each cluster, then contracting these MSTs each to a vertex, and finally finding a MST of the resulting graph. This favourable behaviour is an exception, however, as the next cases show.

A well-known clustered variant of the \emph{traveling salesperson problem}
is that in which one has to find a minimum-cost Hamiltonian cycle of $G$ (where
$G$ is assumed to be complete, and $w$ is assumed to be a metric on $G$) such
that \emph{all} the vertices of each cluster are visited consecutively. For this
problem, Bao and Liu give in \cite{BL12} a 13/6-approximation algorithm, thus
improving a previous approximation ratio of $2.75$ due to Guttmann-Beck \emph{et
al.}  \cite{GHKR00}. As a comparison, recall that the best old-standing
approximation ratio for the unclustered version of the problem is equal to $3/2$
(i.e., the celebrated Christofides algorithm).

Another prominent clustered variant is that concerned with the classic
\emph{minimum Steiner tree problem}. In this case, one has to find a tree
of minimum cost spanning a subset $R \subseteq V$ of \emph{terminal} vertices,
under the assumption that nodes in $R$ are partitioned into a set of clusters,
say $\{R_1,R_2, \ldots, R_k\}$, and with the additional constraint that, in a
feasible solution $T$, we have that, for every $i=1,2,\dots,k$, the minimal
subtree of $T$ spanning the vertices of $R_i$ does not contain any terminal
vertex outside $R_i$. For this problem, again restricted to the case in which
$G$ is complete and $w$ is a metric on $G$, in \cite{WL15} the authors present a
$(2+\rho)$-approximation algorithm, where $\rho \simeq 1.39$ is the best known
approximation ratio for the minimum Steiner tree problem \cite{BGRS13}.

Further, we mention the clustered variant of the \emph{minimum routing-cost
spanning tree problem}. While in the non-clustered version one has to find a
spanning tree of $G$ minimizing the sum of all-to-all tree distances, and the
problem is known to admit a \ptas \cite{WLBCRT98}, in \cite{LW16} the authors
analyze the clustered version, and show that on general graphs the problem is
hard to approximate within any polynomial factor, while if $G$ is complete and
$w$ is a metric on $G$, then the problem admits a factor-2 approximation.
Interestingly, along the way the authors present an inapproximability result for
\clusp (on weighted graphs), which was in fact the inspiration for the present study.

Finally, we refer the reader to the paper by
Feremans \emph{et al.} \cite{FLL03}, where the authors review several classic
network design problems in a clustered perspective, but with different side
constraints on the clusters (i.e., expressed in terms of number of
representatives for each cluster that has to belong to a feasible solution). A
notable example of this type is the so-called \emph{group Steiner tree problem},
where it is required that \emph{at least} one terminal vertex from each cluster
$R_i$ must be included in a feasible solution. This problem is known to be
approximable within $O(\log^{3}n)$ \cite{GKR00}, and not approximable within
$\Omega(\log^{2-\epsilon}n)$, for any $\epsilon>0$, unless \np admits
quasipolynomial-time Las Vegas algorithms \cite{HK03}.

\subsection{Structure of the Paper.}
For the sake of clarity, we first present, in Section~\ref{sec:cluBFS}, our
results on \clubfs. Then, in Section~\ref{sec:cluSPT}, we give our results
on \cluspt, while in Section \ref{sec:cluSP} we
discuss our results on unweighted \clusp. Finally, in Section~\ref{sec:concl}
we conclude the paper and outline possible future research directions.

\section{\clubfs} \label{sec:cluBFS}
In this section, we present our results on the \clubfs problem. In particular,
we first prove that it is \np-hard, then we show that it can be approximated
within an $O(n^{\frac{2}{3}})$-factor in polynomial time, by providing a
suitable approximation algorithm, and finally we show it is fixed-parameter tractable.
We start by proving the following result:

\begin{theorem}
\label{th:clubfshard}
\clubfs is \np-hard.
\end{theorem}
\begin{proof}
In order to prove the statement we provide a polynomial-time reduction from the \tsat
problem, which is known to be \np-complete, to \clubfs.

The \tsat problem is a variant of the classic \sat
problem. \sat is the problem of determining whether it is
satisfiable a given \emph{boolean CNF formula}, i.e., a conjunction of
\textit{clauses}, where a clause is a disjunction of \textit{literals}, and a
literal represents either a \textit{variable} or its negation. In the \tsat
version, the number of literals in each clause is constrained to be exactly
three.

The proof proceeds as follows: starting from a \tsat instance $\phi$ with $\eta$
variables, say $x_1,\dots,x_\eta$, and $\mu$ clauses, say $c_1,\dots,c_{\mu}$,
we first construct an instance $\langle \phigraph,\V,s\rangle$ of the \clubfs
problem which consists of: (i) a graph \phigraph; (ii) a clustering $\V$ of the
vertices of \phigraph; (iii) a distinguished source vertex $s$ of \phigraph.
We then show that instance $\langle \phigraph,\V,s\rangle$ exhibits the two
following properties:
(i) if $\phi$ is satisfiable then $\opt \le 3\eta + 8\mu $;
(ii) if $\phi$
is not satisfiable then $\opt \geq 3\eta + 8\mu + 3$,
where $\opt$ denotes the
cost of the optimal solution to the \clubfs problem on $\langle
\phigraph,\V,s\rangle$.
By proving the above, we will show that finding an optimal solution to the
\clubfs problem is at least as hard as solving \tsat.

The graph \phigraph corresponding to the formula $\phi$ can
be obtained from an empty graph by proceeding as follows.
First, we add to \vphi a single source vertex $s$ and for each variable $x_i$, we add: (i) two \emph{variable
vertices} $v_i$ and $\overline{v}_i$ to \vphi; (ii) three edges, namely
$(v_i,s)$, $(\overline{v}_i,s)$ and $(v_i,\overline{v}_i)$, to \ephi. Then, for each clause $c_j$ we add: (i) three \emph{clause vertices},
$c_{j,1}$, $c_{j,2}$, $c_{j,3}$, one for each of the three literals of $c_j$, to
\vphi; (ii) three edges $(c_{j,1},c_{j,2})$, $(c_{j,2},c_{j,3})$,
$(c_{j,3},c_{j,1})$ to \ephi. Finally, for each clause $c_j$, and for $k=1,2,3$, let $x_i$ be the variable associated with the $k$-th literal $\ell$ of $c_j$. If the literal is negative, i.e., $\ell=\overline{x_i}$, we add edge $(c_{j,k},\overline{v}_i)$ to \ephi, otherwise (i.e., $\ell=x_i$) we add $(c_{j,k},{v_i})$ to \ephi.

It is easy to see that \phigraph has $|\vphi|=3\mu+2\eta+1$  vertices and
$|\ephi|=6\mu+3\eta$ edges.
A clarifying example on how to build \phigraph for a generic \tsat formula $\phi$ having three clauses is shown in Fig.~\ref{fig:reduction_bfs}. Notice that, the first literal of clause $c_1$ is positive and
associated with variable $x_1$. Therefore, clause vertex $c_{1,1}$ is connected to variable vertex $v_1$ in \phigraph.

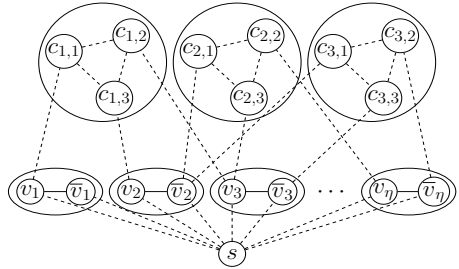
\begin{figure}[ht]
\centering
\resizebox{.5\textwidth}{3.5cm}{
\begin{tikzpicture}[font=\huge,semithick,transform shape,minimum size=8mm,inner sep=0pt]
\node[draw, circle] (s) at (6,0) {$s$};
\node[draw, circle] (1) at (0,2) {$v_1$};
\node[draw, circle] (1n) at (1.5,2) {$\overline{v}_1$};
\node[draw, circle] (2) at (3,2) {$v_2$};
\node[draw, circle] (2n) at (4.5,2) {$\overline{v}_2$};
\node[draw, circle] (3) at (6,2) {$v_3$};
\node[draw, circle] (3n) at (7.5,2) {$\overline{v}_3$};

\node[circle] (dots) at (9,2) {$\cdots$};
\node[draw, circle] (n) at (10.5,2) {$v_\eta$};
\node[draw, circle] (nn) at (12,2) {$\overline{v}_\eta$};

\node[draw,circle] (c11) at (1,6.5) {$c_{1,1}$};
\node[draw,circle] (c12) at (3,7) {$c_{1,2}$};
\node[draw,circle] (c13) at (2.5,5) {$c_{1,3}$};

\node[draw,circle] (c21) at (5,6.5) {$c_{2,1}$};
\node[draw,circle] (c22) at (7,7) {$c_{2,2}$};
\node[draw,circle] (c23) at (6.5,5) {$c_{2,3}$};

\node[draw,circle] (c31) at (9,6.5) {$c_{3,1}$};
\node[draw,circle] (c32) at (11,7) {$c_{3,2}$};
\node[draw,circle] (c33) at (10.5,5) {$c_{3,3}$};
\draw (0.75,2) ellipse (1.4cm and .75cm);
\draw (3.75,2) ellipse (1.4cm and .75cm);
\draw (6.75,2) ellipse (1.4cm and .75cm);
\draw (11.25,2) ellipse (1.4cm and .75cm);
\draw (2.15,6.15) ellipse (1.9cm and 1.9cm);
\draw (6.15,6.15) ellipse (1.9cm and 1.9cm);
\draw (10.15,6.15) ellipse (1.9cm and 1.9cm);
\draw[dashed]   (s) edge (1);
\draw[dashed]   (s) edge (1n);
\draw[dashed]   (s) edge (2);
\draw[dashed]   (s) edge (2n);
\draw[dashed]   (s) edge (3);
\draw[dashed]   (s) edge (3n);
\draw[dashed]   (s) edge (n);
\draw[dashed]   (s) edge (nn);
\draw[dashed]   (c11) edge (c12);
\draw[dashed]   (c12) edge (c13);
\draw[dashed]   (c11) edge (c13);
\draw[dashed]   (c21) edge (c22);
\draw[dashed]   (c22) edge (c23);
\draw[dashed]   (c21) edge (c23);
\draw[dashed]   (c31) edge (c32);
\draw[dashed]   (c32) edge (c33);
\draw[dashed]   (c31) edge (c33);
\draw   (1) edge (1n);
\draw   (2) edge (2n);
\draw   (3) edge (3n);
\draw   (n) edge (nn);
\draw[dashed]   (1) edge (c11);
\draw[dashed]   (2) edge (c13);
\draw[dashed]   (3) edge (c12);
\draw[dashed]   (2n) edge (c21);
\draw[dashed]   (3) edge (c23);
\draw[dashed]   (n) edge (c22);
\draw[dashed]   (2n) edge (c31);
\draw[dashed]   (3n) edge (c33);
\draw[dashed]   (nn) edge (c32);
\end{tikzpicture}}
\caption{Graphical representation of the reduction from \tsat to \clubfs used in the proof of Theorem~\ref{th:clubfshard}.}
\label{fig:reduction_bfs}%
\end{figure}

Now, the final step of the construction consists in defining a clustering
$\V$ over the vertices of \phigraph.
In details, we define
$\V=\{\cluster{s},\cluster{1},\dots,\cluster{\mu},\cluster{\mu+1},
\cluster{\mu+\eta}\}$ as follows. The source vertex is a singleton, i.e.,
$\cluster{s}$ contains $s$ only.
Then, for
each clause $c_j$, with $j=1,\dots, \mu$, we set $\cluster{j} = \{
c_{j,1}, c_{j,2}, c_{j,3} \}$.
Finally, for each variable $x_i$ we set $\cluster{\mu+i}= \{ \overline{v}_i,
v_i \}$.

We now proceed with the last part of the proof. In particular, if $\phi$ is
satisfiable, we consider a satisfying assignment and we construct a solution
$T$ to \clubfs on instance $\langle \phigraph,\V,s\rangle$ as follows:
i) for each variable $x_i$, if $x_i$ is true we add the edges $(s,v_i)$ and $(v_i,
\overline{v}_i)$ to $T$, otherwise we add the edges
$(s,\overline{v}_i)$ and $(v_i, \overline{v}_i)$ to $T$;
ii) for each clause  $c_j$, choose $k \in \{1,2,3\}$ so that the $k$-th
literal of $c_j$ is true, and let $v_i$ be the unique variable vertex that is a
neighbor of $c_{j,k}$ in $G_\phi$. We add the edges in $\{(c_{j,k}, v_i) \} \cup
\{ (c_{j,k}, c_{j,k'}) \, : \, k' \in \{1,2,3\} \wedge k'\neq k\}$ to
$T$.

It is easy  to check that exactly one of each pair of vertices $v_i$ and
$\overline{v}_i$ is at distance $1$ from $s$ in $T$ while the other
is at distance $2$. Moreover, for each clause $c_j$ exactly one of the
vertices in $\{ c_{j,1}, c_{j,2}, c_{j,3} \}$ is at distance $2$ from $s$ in
$T$, while the other two are at distance $3$. Hence $\opt \le 3\eta +
8\mu$.
Suppose now that $\phi$ is not satisfiable and let $T$ be a solution to
\clubfs. It is easy to see that, for each variable $x_i$, solution
$T$ must include the
edge $(v_i,\overline{v}_i)$ since the graph induced by the associated cluster
must be connected. This means that at least one of $v_i$ and $\overline{v}_i$
must be at distance $2$ from $s$ in $T$. Similarly, since for every
$j=1,\dots,\mu$
the subgraph of $T$ induced by the vertices in $\{c_{j,1}, c_{j,2},
c_{j,3}\}$
must be connected, we have that one of them, say w.l.o.g.\ $c_{j,1}$, must be at
a distance at least $d_{T}(s,c_{j,1}) \ge
d_{\phigraph}(s,c_{j,1})=2$ from $s$ while the
other two must be at a distance of at least $d_{T}(s,c_{j,1})+1$ in
$T$. Moreover, since
$\phi$ is not satisfiable, there is at least one clause $c_j$ with $j \in \{1,\dots,\mu\}$ such
that the closest vertex of $c_{j,1}$, $c_{j,2}$, $c_{j,3}$ is at distance at least
$3$ from $s$ in $T$. Indeed, if that were not the case, this would
imply that,
for each clause $c_j$, there would exist a vertex $c_{j,k}$, for a certain
$k$, at distance $2$ from $s$
in $T$, and hence the set of vertices at distance $1$ from $s$ would
induce a
satisfying truth assignment for $\phi$. It follows that:
\begin{align*}
 \cost(T) & \ge \eta + 2\eta + (\mu-1)2 + 2(\mu-1)3 + 3+2\cdot 4\\
 &= 3\eta + 8(\mu-1) + 11 \\
 & = 3\eta + 8\mu +3.
\end{align*}
Since the latter bound holds for any solution to \clubfs, we have $\opt \ge 3\eta + 8\mu +3$, which concludes the proof. \qed
\end{proof}

\subsection{An approximation algorithm}
In what follows, we present an approximation algorithm for
\clubfs (see Algorithm~\ref{alg:clustered_bfs_apx}). The main idea of the
algorithm is that of minimizing the number of distinct clusters that must be
traversed by any path connecting the source $s$ to a vertex $v \in V$. Recall that the
\emph{diameter} $\diam(G)$ of a graph $G$ is the length of a longest
shortest path in $G$. Then, it is possible to show that: (i) if all the clusters
are of low diameter, then this leads to a good approximation for \clubfs, and,
on the other hand (ii) if at least one cluster has large diameter, then the
optimal solution must be expensive and hence any solutions for \clubfs will
provide the sought approximation.

Given an instance $\langle G,\V,s\rangle$ of
\clubfs, w.l.o.g.  let us assume that \cluster{1} is the cluster containing vertex $s$, and
that $G[\cluster{i}]$ is connected for each $i=1,\dots,k$, as otherwise the problem trivially admits no feasible solution.
Our approximation algorithm first
considers each cluster $\cluster{i} \in \V$ and identifies all the vertices
belonging to
$\cluster{i}$ into a single \emph{cluster-vertex} $\cluvertex{i}$ to obtain a
graph $G'$ in which: (i)
each vertex corresponds to a cluster; (ii) there is an edge $(\cluvertex{i},
\cluvertex{j})$ between two vertices in $G'$ if and only if the set $E_{i,j} =
\{(v_i, v_j) \in E : v_i \in \cluster{i}\, \wedge\, v_j \in \cluster{j}\,
\wedge\, i\neq j\}$ is not
empty.
The algorithm proceeds then by computing a BFS tree \clustbfstree of $G'$ rooted
at $\cluvertex{1}$ and constructs the sought approximate solution \approxsol as
follows: initially, \approxsol contains all the vertices of $G$ and the edges
of a BFS
tree of $G[\cluster{1}]$ rooted at $s$. Then, for each edge $(\cluvertex{i},
\cluvertex{j})$ of \clustbfstree, where $\cluvertex{i}$ is the parent of
$\cluvertex{j}$ in \clustbfstree, it adds to \approxsol a single edge $(v_i,r_j)
\in E_{i,j}$ along with all the edges of a BFS tree $T_j$ of $G[\cluster{j}]$
rooted at $r_j$.

\begin{algorithm}[th]
\DontPrintSemicolon
\SetKwInOut{Input}{Input}
\SetKwInOut{Output}{Output}
\Input{An instance $\langle G,\V,s\rangle$ of \clubfs}
\Output{An approximated clustered BFS tree
$\approxsol$}
Let $s \in \cluster{1}$ w.l.o.g. \;
$G' \gets $ Copy $G$ and identify the vertices belonging to  $\cluster{i}$ into a
single cluster-vertex \cluvertex{i}\;
Let $E_{i,j} = \{(v_i, v_j)
\in E : v_i \in \cluster{i} \wedge v_j \in \cluster{j} \}$\;
$\clustbfstree \gets$
Compute a BFS tree of $G'$ rooted at \cluvertex{1}\;
$T_1 \gets $ Compute a BFS tree of $G[\cluster{1}]$ rooted at s\;
$\approxsol \gets (V, E({T_1}) )$ \;
\For{$j=2,\dots,k$}{
	$(p_i, r_j) \gets $ any edge in $E_{i,j}$ where $\cluvertex{i}$ is the
parent of $\cluvertex{j}$ in \clustbfstree \label{alg:clustered_bfs_apx:liner}\;
	$T_j \gets $ Compute a BFS tree of $G[\cluster{j}]$ rooted at $r_j$\;
	$E({\approxsol}) \gets E({\approxsol}) \cup \{ (p_i, r_j) \} \cup
E({T_j})\}$\;
}
\Return \approxsol
\caption{Approximation algorithm for \clubfs.}
\label{alg:clustered_bfs_apx}
\end{algorithm}

We now show that Algorithm~\ref{alg:clustered_bfs_apx} outputs a feasible
solution for the \clubfs problem which is far from the optimum by at most a factor of $\min\{\frac{4n k}{\gamma},\frac{4n^2}{\gamma^2},2\gamma\}$, where $\gamma = \max_{\cluster{i} \in \V}{\diam(G[\cluster{i}])}$.
In particular, given an instance $\langle
G,\V,s\rangle$ of \clubfs, let \optsol be an optimal clustered BFS tree.
To prove the approximation ratio, we will make use of the following lemma.

\begin{lemma}
\label{lemma:clustered_bfs_ub_cost}
$\cost(\approxsol) \le 2 \gamma \, \cost(\optsol)$.
\end{lemma}
\begin{proof}
We first prove that for every $v \in V$ it holds that $ d_{\approxsol}(s,v)
\le \gamma(d_{\optsol}(s, v)+1)$. In particular, let us assume that
$\cluster{i}$ is
the cluster of $\V$ containing $v$. Moreover, let \optsolident be the tree
obtained
from \optsol by identifying each cluster $\cluster{j} \in \V$ into a single
cluster-vertex \cluvertexopt{j}. Observe that $r_i$ is the vertex chosen by
Algorithm~\ref{alg:clustered_bfs_apx} at line~\ref{alg:clustered_bfs_apx:liner}
w.r.t. the cluster $\cluster{i}$ containing $v$.
Therefore, we have that:
\begin{align*}
d_{\approxsol}(s, v) & =  d_{\approxsol}(s, r_i) + d_{\approxsol}(r_i, v) \le  \gamma \, d_{\clustbfstree}(s, \cluvertex{i}) + \gamma\\
& \le   \gamma \,d_{\optsolident}(s, \cluvertexopt{i}) + \gamma  \le  \gamma \, d_{\optsol}(s,v) + \gamma \\&= \gamma \, (d_{\optsol}(s, v)+1),
\end{align*}
from which it follows:
\begin{align*}
\cost(\approxsol) & = \sum_{v \in V} d_{\approxsol}(s, v) \le \sum_{v \in
V} \gamma \, ( d_{\optsol}(s, v)+1)\\& \le \gamma \sum_{v \in V} d_{\optsol}(s,v) + \gamma \, n
\end{align*}
and therefore $\cost(\approxsol)\le \gamma \, \cost(\optsol) + \gamma \, n \le 2 \gamma \,\cost(\optsol)$.
 \qed
\end{proof}

Given the above lemma, we are now ready to prove the following theorem.

\begin{theorem}
Algorithm~\ref{alg:clustered_bfs_apx} is a polynomial-time $\rho$-approximation
algorithm for \clubfs, where $\rho=\min\{\frac{4n k}{\gamma},\frac{4n^2}{\gamma^2},2\gamma\}$.
\end{theorem}
\begin{proof}
First of all, note that there is a least one cluster $\cluster{i}$
such that $\diam(G[\cluster{i}])=\gamma$, and hence it follows that $\cost(\optsol) \ge
\frac{\gamma^2}{4}$. Indeed, if $\gamma$ is even, we have that an optimal solution must pay at least the cost of two paths rooted at the center of a diametral path, namely
\[\cost(\optsol) \ge 2\sum_{i=1}^{\gamma/2} i =  \frac{\gamma^2}{4}+\frac{\gamma}{2}.\]
Similarly, if $\gamma$ is odd, we have
\[\cost(\optsol)  \ge
\sum_{i=1}^{(\gamma-1)/2} i + \sum_{i=1}^{(\gamma+1)/2} i =  \frac{\gamma^2}{4}+\frac{\gamma}{2}+\frac{1}{4}.\]
Now, we observe that $\cost(\approxsol)$ is upper bounded by: \begin{itemize}\item[(i)] $\gamma n k$, since in any feasible solution $T$ to \clubfs, it holds that $d_T(s,v) \leq \gamma k,~\forall v \in V$; \item[(ii)] $n^2$, since $d_G(s,v) \leq n,~\forall v \in V$.\end{itemize}
Therefore, since $\cost(\optsol)\ge \frac{\gamma^2}{4}$, the approximation ratio achieved by Algorithm~\ref{alg:clustered_bfs_apx} is always upper bounded by $\min\{\gamma n k,n^2\} \cdot \frac{4}{\gamma^2} = \min\{\frac{4 n k}{\gamma},\frac{4 n^2}{\gamma^2}\}$.
Moreover, by Lemma~\ref{lemma:clustered_bfs_ub_cost} we also know that $\cost(\approxsol) \le 2 \gamma \, \cost(\optsol)$. Hence, overall, Algorithm~\ref{alg:clustered_bfs_apx} always computes a solution \approxsol such that $\frac{\cost(\approxsol)}{\cost(\optsol)} \le \rho$, where $\rho=\min\{\frac{4n k}{\gamma},\frac{4n^2}{\gamma^2},2\gamma\}$. 
Since the time complexity is upper bounded by the cost of computing the BFS trees, the claim follows. \qed
\end{proof}

Notice that each of the three terms in $\rho$ can be the minimum one, depending on the structure of a given instance of \clubfs.
In particular, the first term is the unique minimum when $\sqrt{2nk} < \gamma < \frac{n}{k}$, and the considered interval is not empty, i.e., when $\sqrt{2nk}<\frac{n}{k}$, which implies $k < \sqrt[3]{n/2}$.
In this latter case, the second term is to be preferred when $\gamma > \frac{n}{k}$, while the minimum is attained by the third term when $\gamma < \sqrt{2nk}$.
Otherwise, i.e., when $k \ge \sqrt[3]{n/2}$, and hence the aforementioned interval is empty, then the second (resp., third) term is the unique minimum when $\gamma$ is larger (resp., smaller) than $\sqrt[3]{2 n^2}$.
Remarkably, when $\gamma$ is either $O(1)$ or $\Theta(n)$, our algorithm thus provides a $O(1)$-approximation ratio.
Finally, notice that if we set $\gamma = \Theta(n^\frac{2}{3})$ and $k=\Theta(n^\frac{1}{3})$, then the three terms in $\rho$ coincide and are equal to $\Theta(n^\frac{2}{3})$, which then happens to be the achieved ratio of our approximation algorithm in the worst case.

\subsection{Fixed-Parameter Tractability} \label{ssec:FPT-CluBFS}
In this subsection, we prove that \clubfs is \emph{fixed-parameter tractable}
(FPT) w.r.t. two natural cluster-related parameters, by providing two different
FPT algorithms, namely \firstfpt and \secondfpt. Recall that an FPT algorithm is
allowed to have an exponential running time, but only in terms of some natural
parameter of the problem instance that can be expected to be small in typical
applications.

\subsubsection{Algorithm \firstfpt}
In \firstfpt we choose as our first natural parameter the number $k$ of clusters of $\V$.
Notice that every feasible solution $T$ for \clubfs induces a \emph{cluster-tree}
\clustertree obtained from $T$ by identifying the vertices of the same
cluster into a single vertex. The main idea underlying the algorithm is that of guessing, for
each vertex of an optimal cluster-tree $\optclustertree$, the vertices belonging to the subtrees rooted in (one of)
its children and then to iteratively reconstruct $\optclustertree$.
For the sake of simplicity, in the following we will assume that $n$ is a power of two. However, note that this assumption can be removed by either modifying the input graph or by tweaking the definition of the functions $f_{v,i}$ and $g_{v,i}$ that are given later in this subsection.

We start with some definitions. First, given an instance $\langle G, \V, v\rangle$ of \clubfs, for any $\cluster{i} \in \V$, and $v \in
\cluster{i}$, we let $\bfs{\cluster{i}}[v]$ be the cost of a BFS tree of
$G[\cluster{i}]$ having source vertex $v$, i.e.,
\[
\bfs{\cluster{i}}[v] = \sum_{u \in \cluster{i}} d_{G[\cluster{i}]}(v,u).
\]

Then, we define a set $U=\V \cup A$ as the union of the set of clusters $\V$ with a set $A=\{a_1, \dots, a_{\log n}\}$, containing $\log n$ additional elements. Moreover, we let $\mu : 2^A \to V$ be a bijection that maps each of the $2^{\log n}=n$ subsets of $A$ to a  vertex of $V$.

For each $\varset \subseteq U$, we let $\opt_{v, i}[\varset]$ be a quantity
depending on the cost $c^*$ of an optimal solution to an auxiliary instance $\langle G', \varset \cap \V, v\rangle$ of \clubfs, where $G'$ is the subgraph of $G$ induced by the vertices in $\cup_{C \in \varset \cap \V} C$, provided that the following constraints are all satisfied:

\begin{enumerate}
\item[(i)] $|\varset \cap \V | \le i$ (i.e., we restrict to subproblems having at most $i$ clusters);
\item[(ii)] $v \in \cup_{C \in \varset \cap \V} C$;
\item[(iii)] $G'$ is connected;
\item[(iv)] $A \subseteq \varset$.
\end{enumerate}

Let $M$ be a parameter whose value will be specified later.
If all the above mentioned four constraints are satisfied and $c^* < M$, then we define $\opt_{v, i}[\varset]=c^*$.
Otherwise, if (i) is satisfied and either $c^* \ge M$ or at least one of (ii)--(iv) is not satisfied, then we allow $\opt_{v, i}[\varset]$ to be any value larger than or equal to $M$.
Finally, if (i) is not satisfied, then we allow  $\opt_{v,i}[\varset]$ to be any upper bound to $c^*$.

Therefore, according to our definition, we have that, whenever (i), (ii), (iii), and (iv) are satisfied, we can set:
\begin{equation}
\label{eq:fpt1_modified_base}
\opt_{v,1}[\varset] = \min\{\bfs{\varset \cap \V}[v], M \}
\end{equation}
Otherwise we set $\opt_{v,1}[\varset]=M$.

Now, let $\eta(\varset \cap \V) = \sum_{C \in \varset \cap \V} |C|$ be
the number of vertices in the clusters of $\varset \cap \V$. Moreover, given a vertex $v \in V$, let \[
\ell(v,v') = \min_{ \substack{(x,v') \in E(G) \\ x \in V(G[R])}} \{ d_{G[R]}(v,x) + 1 \},
\]
where $R \in \varset \cap \V$ is the cluster containing $v$ and $\ell(v,v')$ is the shortest among the paths from $v$ to $v'$ that
traverse only vertices in $R$, except for $v'$. If there is no such path, then $\ell(v,v') = +\infty$.

Hence, for $i>1$ we can write the following recursive formula:
\begin{equation}
\label{eq:fpt1}
\begin{split}
 \opt_{v,i}[\varset] &= \min_{ \varsetprime \subseteq \varset } \big\{
	L(v, \mu(\varsetprime \cap A), \varsetprime \cap V)
	\\&+ \opt_{\mu(\varsetprime \cap A), i-1}[ (\varsetprime \cap \V) \cup A ]
	\\&+ \opt_{v, i-1}[  (\varset \setminus \varsetprime) \cup A \big\},
\end{split}
\end{equation}
where $L(v, \mu(\varsetprime \cap A), \varsetprime \cap V) = \min\{\ell(v, v') \eta(\varsetprime \cap \V), M\}$ accounts for (the lengths of) the portions of the shortest paths from $v$ to the vertices in $C' = \cup_{C \in \varsetprime \cap V} C$  whose edges are not in the subgraph induced by $C'$.

Given the above formula, we now show that $\opt_{v, i}[\varset]$ for $i>1$ can be computed efficiently by exploiting a result provided in~\cite{BHKK07}, namely the following:
\begin{theorem}[\cite{BHKK07}]
Given a set $X$ and two functions $f,g : 2^X \to [-W, \dots, W]$, it is possible to compute in $\mathtt{conv}(W,X):= O(W \cdot |X|^3 \cdot 2^{|X|} \cdot \polylog(W,|X|))$ time\footnote{The runtime originally given in~\cite{BHKK07} is here restated on our (implicitly assumed) model of computation, namely the standard unit-cost RAM with logarithmic word size, on which the $O(|X|^2 \cdot 2^{|X|})$ ring operations performed in~\cite{BHKK07} cost $O(W \cdot |X| \cdot \polylog(W,|X|))$ time each.
Notice that we are explicitly stating polynomial factors in $|X|$, i.e., logarithmic factors in $2^{|X|}$, which are disregarded in~\cite{BHKK07},   since they will result in polynomial factors in $k$ in the running time of our FPT algorithm.} the \emph{subset convolution} $(f*g)$ of $f$ and $g$ over the min-sum semiring, i.e., for every set $Y \subseteq X$ the quantity:
\[
	(f * g)(Y) = \min_{Z \subseteq Y} \{  f(Z) + g(Y \setminus Z)\}.
\]
\label{th:bjork}
\end{theorem}

In particular, the main idea here is to express the values $\opt_{v, i}[\varset]$ as subset convolutions of two suitable functions $f$ and $g$.
In more details, notice that Equation~\eqref{eq:fpt1} can be rewritten as follows:
\begin{equation}
	\label{eq:fpt1_convolution}
	\opt_{v,i}[\varset] = \min_{ \varsetprime \subseteq \varset } \{ f_{v,i}(\varsetprime) + g_{v,i}(\varset \setminus \varsetprime) \} = (f_{v,i} * g_{v,i})(\varset)
\end{equation}

\noindent
once we define
\begin{align*}
	f_{v,i}(X) &= L(v, \mu(X \cap A), X \cap V) + \opt_{\mu(X \cap A), i-1}[ (X \cap \V) \cup A ], \mbox{ and}\\
	g_{v,i}(X) &= \opt_{v, i-1}[ X \cup A ].
\end{align*}

Notice also that, if we interpret $M$ to be an upper bound to the cost of any optimal solution of the original \clubfs instance (e.g., by selecting $M=n^2$), then we have that $\opt_{v, i}[S \cup A]$, for every $i \ge |S|$ and for any $S \subseteq \V$, coincides with the cost of the optimal solution to the instance $\langle G', S, v\rangle$ of \clubfs whenever such an instance is feasible. Otherwise, we have that $\opt_{v, i}[S \cup A]$ is at least $M$.

Hence, the above relation can be exploited to define the following algorithmic process. We start by choosing $M=1$ and then we perform a series of rounds as follows.
In each round, we first determine all the values $\opt_{v,1}[\varset]$ by using Equation~\eqref{eq:fpt1_modified_base}. Then, for every $i=2, \dots, k$, we compute $n$ subset convolutions as shown in~Equation\eqref{eq:fpt1_convolution} (using Theorem~\ref{th:bjork}).
In more details, we compute $f_{v,i} * g_{v,i}$ of Equation~\eqref{eq:fpt1_convolution} for each vertex $v \in V$.

Finally, we set $\opt_{v,i}[\varset] = \min\{ (f_{v,i} * g_{v,i})(\varset), M \}$ and we move to the next iteration. Here the minimum is necessary in order to ensure that the values computed by the subset convolutions that rely on $\opt_{v,i}[\varset]$ will be in $O(M)$.
After the last iteration of this round is completed, $\opt_{s,k}[\V \cup A]$ stores either $M$ or a value strictly smaller than $M$.
%
On the one hand, if $\opt_{s,k}[\V \cup A]<M$, we have found the cost $\opt$ of an optimal solution of the original instance, i.e., %
$\opt = \opt_{v, k}[\V \cup A]$.
The optimal tree \optclustertree can then be reconstructed from the values
$\opt_{v, k}[S \cup A]$ for any $S \subseteq \V$, by using, e.g., the method in~\cite{DPV06}.
On the other hand, if $\opt_{s,k}[\V \cup A] = M$, we move to the next round: we double the value of $M$ and repeat the above procedure.
We are now ready to give the following result.

\begin{lemma}
\label{lm:clubfs_fpt1_modified}
\clubfs can be solved in $\softO(2^k k^3 n^4)$ time.
\end{lemma}
\begin{proof}
First of all, notice that the cost of all the BFS trees of the clusters in $C \in \V$, from all the vertices $v \in V$, can be computed in $\softO(nm)$ time. 
Hence, it follows that all the $n \cdot 2^{|U|} = n \cdot 2^{k + \log n} = n^2 \cdot 2^k$ base cases $\opt_{v,1}[\varset]$ can be computed in $O(n^3 + n^2 \cdot 2^k)$ time.
	
Now we focus on the values $\opt_{v,i}[\varset]$ having $i>1$. In particular, notice that, for each $M$ considered in the process, since functions $f_{v,i}$ and $g_{v,i}$ have values between $0$ and $2M$, we can compute all $n$ values $\opt_{v,i}[\varset]$ for each $\varset \subseteq U$, in

\begin{align*}
n \cdot \mathtt{conv}(2M,U) &= n \cdot O(2M \cdot |U|^3 \cdot 2^{|U|} \cdot \polylog(2M,|U|)) \\
&=n \cdot O(M \cdot (k+\log n)^3 \cdot 2^{k+\log n} \cdot \polylog(M,k + \log n)) \\
&=n \cdot O(M \cdot k^3 \cdot n \cdot 2^k \cdot \polylog(n^2, n + \log n))=
\softO(M \cdot 2^k \cdot k^3 \cdot n^2).
\end{align*}
\noindent time.
Overall, we perform at most $1 + \left\lceil \log \opt \right\rceil$ rounds, since we stop as soon as $M > OPT$ (i.e., when
we have $M>\opt_{s,k}[\V \cup A] = \opt$).
Hence, the overall time complexity of all rounds is
\begin{align*}
 & = \softO\left(\sum_{M=1,2,4,\dots,2\opt} M \cdot 2^k \cdot k^3 \cdot n^2 \right)\\
 & = \softO(\opt \cdot 2^k \cdot k^3 \cdot n^2)\\
 & = \softO(2^k \cdot k^3 \cdot n^4),
 \end{align*}
since $\Theta(n^2)$ is a trivial upper bound on the cost \opt of any feasible solution to \clubfs.
\qed
\end{proof}

It is worth noting that, in realistic settings, the number of clusters depends on various parameters, such as type of deployed devices and network density.
However, it is almost always expected to be a small fraction w.r.t. overall number of vertices (see, e.g.~\cite{FJARQK12,SK08}). Thus, \firstfpt might result in being truly effective in practice.

However, when this is not the case, then its running time might easily become impractical.
In particular, if we focus on the classical BFS tree problem, which can be seen as a special instance of \clubfs where each cluster contains only one vertex, it is easy to see that \firstfpt takes exponential time while the problem is known to be trivially solvable in $O(m+n)$ time! This suggests that, for the case in which $\V$ consists of many singleton clusters, there must be another parametrization yielding a better complexity. Following this intuition, in the remaining of this section we present another FPT algorithm, namely \secondfpt, parameterized in $\param = |\{v\in V : v \in \cluster{i}, \cluster{i} \in \V, |\cluster{i}| > 1 \}|$, i.e., in the total number of vertices that belong to clusters of size at least two.

\subsubsection{Algorithm \secondfpt}
The idea underlying \secondfpt is as follows. Given a solution $T$ to \clubfs
we call a \emph{cluster root} for $\cluster{i}\in \V$ the unique vertex $v\in
\cluster{i}$ with the smallest distance from $s$ in $T$. The \secondfpt
algorithm guesses the root of each cluster in an optimal solution \optsol
and then computes the optimal way of connecting the different roots of the
clusters together.

Suppose we know a vector $\langle v_1, \dots, v_k \rangle$ of vertices such that
$v_i \in \cluster{i}$. The key observation is that we can write the cost of any
solution $T$ having vertices $v_1, \dots, v_k$ as cluster roots as follows:
\begin{align*}
\cost(T) &= \sum_{\cluster{i} \in \V} \sum_{v \in \cluster{i}} d_T(s, v) \\
  &=  \sum_{\cluster{i} \in \V} \left( |\cluster{i}| d_T(s, v_i) +
\sum_{v \in \cluster{i}} d_T(v_i, v) \right) \\
& = \sum_{\cluster{i} \in \V} |\cluster{i}| d_T(s, v_i)  + \sum_{v \in \cluster{i}} d_T(v_i, v).
\end{align*}

Since $d_T(v_i, v) \ge d_{G[\cluster{i}]}(v_i, v)$, for any $v \in \cluster{i}$,
the second summation is minimized when $d_T(v_i, v) = d_{G[\cluster{i}]}(v_i,
v)$, i.e.,
when
$T[\cluster{i}]$ is a BFS tree of $G[\cluster{i}]$.
Consider now the first summation,
and focus on its generic $i$-th term. Let
$\V'$ be the set of clusters traversed by the path $\pi = \pi_T(s,v_i)$. For
each cluster $\cluster{j} \in \V'$, let $x,y \in \cluster{j}$ be the first and
last vertex of $\cluster{j}$ traversed by $\pi$, respectively. By the
definition of \clubfs, and of cluster root, for \cluster{i} we have that: (i)
all the vertices in the subpath of $\pi$ between $x$ and $y$, say $\pi[x,y]$,
belong to $\cluster{i}$ and (ii) $ x = v_i$.
Let $P_i$ be the set of all the paths in $G$ from $s$ to $v_i$ satisfying
conditions (i) and (ii).
It is easy to see that $d_T(s, v_i) \ge \min_{\pi'\in P_i}|\pi'|$. Hence, if $T$
contains, for each $\cluster{i} \in \V$, the shortest path in $P_i$ then
$\sum_{\cluster{i} \in
\V} |\cluster{i}| d_T(s, v_i)$ is minimized.
To determine any path in $P_i$ we proceed as follows. We define an auxiliary
directed graph $G'$, obtained from $G$ by:
(i) removing all the edges $(x,y) \in E$ such that neither $x$ nor $y$ is a
root-vertex $v_i$ for some $i$;
(ii) directing all the edges $(x,y) \in E$ such that $x$ or $y$ is a
root-vertex $v_i$ for some $i$ towards $v_i$; if both $x=v_i$ and $y=v_j$ (for
some $i,j$) then we replace the undirected $(x,y)$ by the pair of directed edges
$(x,y)$ and $(y,x)$;
(iii) replacing, for all $\cluster{i} \in \V$, all the edges in
$E(G[\cluster{i}])$ with the
edges of a BFS tree of $G[\cluster{i}]$ rooted in $v_i$. These edges are
directed from
the root towards the leaves of the tree.
It is easy to see that any path in $P_i$ is contained in $G'$, and that any BFS
tree of $G'$ must contain the edges of all the BFS trees of $G[\cluster{i}]$,
hence
minimizing $\cost(T)$. Therefore the optimal solution to the instance of \clubfs
contains exactly the (undirected version of) the edges of a BFS tree of $G'$.
The following lemma follows from the above discussion.

\begin{lemma}
\label{lm:clubfs_fpt2}
 \secondfpt solves \clubfs in $O( \param^ \frac{\param}{2} m)$ time.
\end{lemma}
\begin{proof}
There are $\prod_{\cluster{i} \in \V} |\cluster{i}|$ ways of
choosing a
vector of cluster root vertices $\langle v_1, \dots, v_k \rangle$ for a given
set \V{} of clusters.
For each of these vectors the algorithm requires a computation of the
BFS trees of $G[\cluster{i}]$ for $i=1,\dots, k$ plus an additional BFS tree of
$G'$.
This can be done in $O(m) + \sum_{i=1}^k O(|E(G[\cluster{i}])|
+ |\cluster{i}| ) = O(m)$ time.
Finally, notice that $\prod_{i=1}^{k} |\cluster{i}| \le \param^\frac{\param}{2}$ as the
total number of clusters of size at least $2$ is at most $\frac{\param}{2}$. \qed
\end{proof}
Since it is possible to show that $\optsol[\cluster{i}]$
must coincide with a BFS tree of $G[\cluster{i}]$ rooted at $r_i$, then this
property allows us to efficiently reconstruct the optimal tree \optsol, for
a given guessed set of roots.
Thus, overall, by combining \firstfpt and \secondfpt, we can give the following result:

\begin{theorem}
\clubfs can be solved in \totalfpt
time.
\end{theorem} 
\section{\cluspt} \label{sec:cluSPT}
In this section, we give our results on the \cluspt problem. In particular,
we first show that \cluspt cannot be approximated, in polynomial
time, within a factor of $n^{1-\epsilon}$ for any constant $\epsilon \in (0,
1]$, unless $\p=\np$. Then, we give an $n$-approximation algorithm, thus
proving that the mentioned inapproximability result is (essentially) tight.
Finally, we show that, similarly to \clubfs, \cluspt is fixed-parameter
tractable.
Since \cluspt is a generalization of \clubfs, Theorem~\ref{th:clubfshard} immediately
implies that \cluspt is \np-hard as well. We can actually provide a stronger
result, namely:

\begin{figure}[ht]
	\centering
	\resizebox{.5\textwidth}{3.7cm}{
		\begin{tikzpicture}[font=\huge,semithick,transform shape,minimum size=8mm,inner sep=0pt]
		\node[draw, circle] (s) at (6,0) {$s$};
		\node[draw, circle] (1) at (0,2) {$v_1$};
		\node[draw, circle] (1n) at (1.5,2) {$\overline{v}_1$};
		\node[draw, circle] (2) at (3,2) {$v_2$};
		\node[draw, circle] (2n) at (4.5,2) {$\overline{v}_2$};
		\node[draw, circle] (3) at (6,2) {$v_3$};
		\node[draw, circle] (3n) at (7.5,2) {$\overline{v}_3$};
		
		\node[circle] (dots) at (9,2) {$\cdots$};
		\node[draw, circle] (n) at (10.5,2) {$v_\eta$};
		\node[draw, circle] (nn) at (12,2) {$\overline{v}_\eta$};

		\node[draw,circle] (c1) at (2.15,5.6) {$r_{1}$};
		\node[draw,circle] (c2) at (6.15,5.6) {$r_{2}$};
		\node[draw,circle] (c3) at (10.15,5.6) {$r_{3}$};
		
		\draw (c1) --(3.15,8.5) -- (1.15,8.5) -- (c1);
		\node[circle] (M1) at (2.15,7.5) {$M$};
		\draw (c2) --(7.15,8.5) -- (5.15,8.5) -- (c2);
		\node[circle] (M2) at (6.15,7.5) {$M$};
		\draw (c3) --(11.15,8.5) -- (9.15,8.5) -- (c3);
		\node[circle] (M3) at (10.15,7.6) {$M$};
		
		\node[draw,circle] (c11) at (1,6) {$c_{1,1}$};
		\node[draw,circle] (c12) at (3.1,6.5) {$c_{1,2}$};
		\node[draw,circle] (c13) at (2.5,4.5) {$c_{1,3}$};
		\node[draw,circle] (c21) at (5,6) {$c_{2,1}$};
		\node[draw,circle] (c22) at (7.1,6.5) {$c_{2,2}$};
		\node[draw,circle] (c23) at (6.5,4.5) {$c_{2,3}$};
		\node[draw,circle] (c31) at (9,6) {$c_{3,1}$};
		\node[draw,circle] (c32) at (11.1,6.5) {$c_{3,2}$};
		\node[draw,circle] (c33) at (10.5,4.5) {$c_{3,3}$};
		
		\draw (0.75,2) ellipse (1.4cm and .75cm);
		\draw (3.75,2) ellipse (1.4cm and .75cm);
		\draw (6.75,2) ellipse (1.4cm and .75cm);
		\draw (11.25,2) ellipse (1.4cm and .75cm);
		\draw (2.15,6.3) ellipse (1.9cm and 2.7cm);
		\draw (6.15,6.3) ellipse (1.9cm and 2.7cm);
		\draw (10.15,6.3) ellipse (1.9cm and 2.7cm);
		
		\draw[dashed]   (s) edge (1);
		\draw[dashed]   (s) edge (1n);
		\draw[dashed]   (s) edge (2);
		\draw[dashed]   (s) edge (2n);
		\draw[dashed]   (s) edge (3);
		\draw[dashed]   (s) edge (3n);
		\draw[dashed]   (s) edge (n);
		\draw[dashed]   (s) edge (nn);
		\draw   (c11) edge (c1);
		\draw   (c12) edge (c1);
		\draw   (c13) edge (c1);
		\draw   (c21) edge (c2);
		\draw   (c22) edge (c2);
		\draw   (c23) edge (c2);
		\draw   (c31) edge (c3);
		\draw   (c32) edge (c3);
		\draw   (c33) edge (c3);
		\draw   (1) edge (1n);
		\draw   (2) edge (2n);
		\draw   (3) edge (3n);
		\draw   (n) edge (nn);
		\draw[dashed]   (1) edge (c11);
		\draw[dashed]   (2) edge (c13);
		\draw[dashed]   (3) edge (c12);
		\draw[dashed]   (2n) edge (c21);
		\draw[dashed]   (3) edge (c23);
		\draw[dashed]   (n) edge (c22);
		\draw[dashed]   (2n) edge (c31);
		\draw[dashed]   (3n) edge (c33);
		\draw[dashed]   (nn) edge (c32);
		\end{tikzpicture}}
	
	\caption{Graphical representation of the reduction used in the proof of Theorem~\ref{th:clusptnonapx}.}%
	\label{fig:reduction_nonapxspt}%
\end{figure}
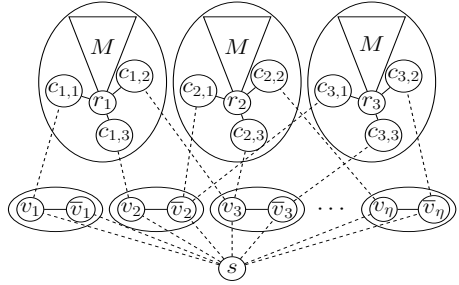

\begin{theorem}
\label{th:clusptnonapx}
\cluspt cannot be approximated, in polynomial time, within a factor
of $n^{1-\epsilon}$ for any constant $\epsilon \in (0, 1]$, unless $\p=\np$.
\end{theorem}
\begin{proof}
To prove the statement we use a slight modification of the construction
given in the proof of Theorem~\ref{th:clubfshard}.
The main difference resides in the structure of the graph
\phigraph. In more details, for each clause $c_j$ we do not add
a triangle of vertices clustered into $\cluster{j}$. Instead, we add a subgraph
to \phigraph which is basically made of two components, as follows. First, we
add four vertices, namely $c_{j,1}$, $c_{j,2}$, $c_{j,3}$ and $r_j$, to \vphi
and connect them in order to form a star graph with center $r_j$. Then, we
create a tree of $M$ vertices, where $M$ is a parameter that will be specified
later, which is connected to the above star graph through the center vertex
$r_j$ only. Finally, we cluster the two components together to form
$\cluster{j}$.
All edges have weight equal to zero, except those that connect the two
vertices associated with a variable, which are unit-weighted.
An example of the modified instance is shown in Fig.~\ref{fig:reduction_nonapxspt}, where the triangle with label $M$ represents a generic
tree of $M$ vertices, rooted, for each clause $j$, at vertex $r_j$.
Now, by using an argument similar to that proposed in the proof of
Theorem~\ref{th:clubfshard}, it is easy to see that instance $\langle
G_{\phi},\V,s\rangle$, defined as above, exhibits the following properties:
(i) if $\phi$ is satisfiable then $\opt = \eta$
(ii) if $\phi$ is not satisfiable then $\opt \geq \eta+M+4$, where $\opt$ denotes the
cost of the optimal solution to the \cluspt problem on instance $\langle
\phigraph,\V,s\rangle$ and $M$ can be chosen as an arbitrarily large integer.

We are now ready to prove the claim. Let $\langle G_{\phi},\V,s\rangle$ be an
instance of \cluspt and let \opt be the cost of an optimal solution to such
instance. Suppose by contradiction that there exists a polynomial-time
$n^{1-\epsilon}$-approximation algorithm $A$ for \cluspt for some constant
$\epsilon \in (0,1]$.
Consider a \tsat instance along with the corresponding \cluspt instance.
W.l.o.g., let us assume that $\mu = \Theta(\eta)$. Note that
\np-hard instances of \tsat of this latter kind are known to exist. We then
set $M = \Theta(\eta^{2/\epsilon})$ so
that the number of vertices of graph $G_\phi$ is $n = \Theta(\mu \cdot M) =
\Theta(\eta^{1+2/\epsilon})$.
If the \tsat instance is satisfiable, then $A$ would return a solution $T$ to
the \cluspt instance having a cost of at most:
$
 \cost(T) \le n^{1-\epsilon} \eta =
O(\eta^{\frac{2-2\epsilon}{\epsilon}} \eta) = O(\eta^{\frac{2}{\epsilon}-1})
= O(M \eta^{-1}) = o(M)$,
while if it is not satisfiable $\cost(T) \ge M$. Hence this would
solve \tsat in polynomial time.  \qed
\end{proof}

\subsection{An approximation algorithm}
We now show that the previous inapproximability result for \cluspt is tight by providing a simple approximation algorithm, as stated in the following.

\begin{theorem}
There exists a polynomial-time $n$-approximation algorithm for \cluspt.
\end{theorem}
\begin{proof}
The algorithm works as follows: first it computes a multigraph $G'$
from $G$ by identifying each cluster $\cluster{i} \in \V$ into a single vertex
\cluvertex{i}. When doing this, it associates each edge of $G'$ with the
corresponding edge of $G$. Then it computes a \emph{minimum spanning tree}
(MST from now on) $T'$ of $G'$, and $k$ MSTs
$T_1, \dots, T_k$ of $G[\cluster{1}], \dots, G[\cluster{k}]$, respectively.
Finally, the algorithm returns the spanning tree \approxsol of $G$ which contains all
the edges in $\overline{E} \cup \cup_{i=1}^k E(T_i)$, where $\overline{E}$
denotes the set of edges of $G$ associated with the edges of $T'$.

Let us now estimate the quality of \approxsol. Let \optsol be an optimal solution
to the \cluspt instance. For a given spanning
tree $T$ of $G$ rooted at $s$, let $w(T)=\sum_{e \in E(T)} w(e)$. Observe that
clearly $w(T) \leq \cost(T) \le n \cdot w(T)$. Moreover, by construction,
$w(\approxsol) \leq w(\optsol)$. Thus, we have:
$\cost(\approxsol) \le  n \cdot w(\approxsol) \le  n \cdot w(\optsol) \le n
\cdot \cost(\optsol)$.
Since the time complexity is upper bounded by the complexity of computing the MSTs, the claim follows.  \qed
\end{proof}

\subsection{Fixed-Parameter Tractability Results}
The fixed-parameter tractability of \cluspt directly follows from the discussion
of Section~\ref{sec:cluBFS} on the FPT algorithms for \clubfs.
In particular, if we focus on \firstfpt, we observe
that it can be trivially adapted to weighted graphs by considering SPTs instead of BFS trees, thus redefining the base cases $\opt_{v,1}[H]$ and the function $\ell(v, v')$.
The only difference in the analysis is that it is no longer possible to use $n^2$ as an upper bound for the value of $M$. However, by retracing the calculations in the proof of Lemma \ref{lm:clubfs_fpt1_modified}, and by using the fact that $M = O(\opt)$, one can easily prove the following:

\begin{lemma}
\cluspt can be solved in $\softO(nm + 2^k k^3 n^2 \cdot \opt \, \log \opt)$ time. 
\end{lemma}

Regarding \secondfpt, it can also be easily adapted to solve \cluspt by using Dijkstra's algorithm
instead of the BFS algorithm, when the solution to a sub-problem defined within each cluster has to be computed. This only slightly increases
the resulting time complexity, which is however in the order of a logarithmic factor, as stated in the following.

\begin{theorem}
\cluspt can be solved in $O\big(\param^\frac{\param}{2}  (m + n \log n)\big)$
time.
\end{theorem}
\begin{proof}
We prove the claim by elaborating on the proofs of Lemma~\ref{lm:clubfs_fpt2}.
In particular, it suffices to note that in Algorithm~\secondfpt, for each vector of cluster root vertices, we need to compute the
SPT trees (instead of BFS trees) of $G[\cluster{i}]$ for $i=1,\dots, k$ plus an additional SPT tree of
$G'$.
This can be done in  $O(m) + \sum_{i=1}^k O(|E(G[\cluster{i}])|+|V(G[\cluster{i}] \log V(G[\cluster{i}])| + |\cluster{i}| ) = O(m + n \log n)$ time. 
\qed
\end{proof}

To summarize, we can give the following theorem.

\begin{theorem}
\cluspt can be solved in \totalfptspt time.
\end{theorem}

\section{\clusp} \label{sec:cluSP}
To complement our results, we also studied \clusp, i.e., the problem of
computing a \emph{clustered shortest path} between two given vertices of a
graph. The problem was introduced in~\cite{LW16}, and asks for finding a
minimum-cost path, in a clustered \emph{weighted} graph $G$, between a source
and a destination vertex, with the constraint that in a feasible path, vertices
belonging to a same cluster must induce a (connected) subpath.
In this section, we extend the results of~\cite{LW16} by considering the
unweighted version of the problem, which to the best of our knowledge was never
considered before this work.
We are then able to give the following result.

\begin{theorem}\label{thm:clusp_non_apx}
Unweighted \clusp cannot be approximated, in polynomial time, within a factor
of $n^{1-\epsilon}$ for any constant $\epsilon \in (0, 1]$, unless $\p=\np$.
\end{theorem}
\begin{proof}
	To prove the statement we show a polynomial-time reduction from the \np-complete problem \xc
	(X3C) to \clusp. In the X3C problem we are given a set
	$\I=\{x_1, \dots, x_{3\eta}\}$ of $3\eta$ items, and a collection $\S = {S_1, \dots,
		S_\mu}$ of $\mu \ge \eta/3$ subsets of $\I$, each containing exactly $3$
	items. The problem consists of determining whether there exists a collection $\S^*
	\subset S$ such that $|\S^*|=\eta$ and $\cup_{S \in \S^*}=\I$ (i.e., each
	element of $\I$ is contained in exactly one set of $\S^*$). For the sake of
	simplicity we assume that each $x_i \in \I$ is contained in at most $3$
	sets.\footnote{The X3C problem remains NP-complete even with this additional
		assumption, see e.g., problem SP2 in \cite{GJ79}.}	
	
	Let $M$ be an integer parameter that will be specified later. Given an instance
	$\langle \I, \S \rangle$ of X3C, the corresponding instance of \clusp is
	constructed as follows:
	\begin{itemize}
		\item For each set $S_j \in \S$ we add four vertices $u_j^0$, $u_j^1$, $u_j^2$, and $u_j^3$.
		\item For $j=1, \dots, \mu-1$ we add the edge $(u_j^3, u_{j+1}^0)$.
		\item For each $x_i \in \I$, let $\ell_i$ be the number of sets that contain $x_i$.
		First we add $\ell_i$ vertices $v_i^z$ for $z=1, \dots, \ell_i$. Then,
		we add a vertex $v_i$ and
		we connect $v_i$ to each $v_i^z$ using a path of length $M$.
		The vertex $v_i$ along with all the vertices in the paths from $v_i$ to each of $v_i^1, \dots, v_i^{\ell_i}$ form a cluster.
		\item  For each $S_j \in \S$ and $x_i \in S_j$, if $x_i$ is the $k$-th item in $S_j$, and $S_j$
		is the $h$-th set to contain $x_i$, we add the edges $(u_j^{k-1}, v_i^h)$ and $(u_j^{k}, v_i^h)$.
		For a given $S_j \in \S$, we call the set of the edges of the form $(u_j^0, v_i^h)$ the
		\emph{top-path for $S_j$}.
		\item For each $z=1,\dots,\mu-\eta$, we add $\mu$ vertices $y_z^1, \dots, y_z^\mu$ and an
		additional vertex $y_z$ connected to each $y_z^1, \dots, y_z^\mu$ with a path of length $M$.
		The vertex $y_z$ along with all the vertices in the paths from $y_z$ to each of $y_z^1, \dots, y_z^\mu$ form a cluster.
		\item For each set $S_j \in \S$ we add the $2(\mu-\eta)$ edges $\{ (u_j^0,
		y_z^j) \, : \, z=1,\dots,\mu-\eta \} \cup \{ (u_j^3, y_z^j) \, : \,
		z=1,\dots,\mu-\eta \}$. For a given $S_j \in \S$, we call the set of edges $\{
		(u_j^0, y_z^j), (u_j^3, y_z^j) \}$ the \emph{$z$-th bottom-path for $S_j$}.
	\end{itemize}
	All vertices that have not explicitly been already assigned to a cluster
	belong to singleton clusters. Moreover we let $s=u_0^0$ and $t=u_\mu^3$.
	An example of the above construction is shown in Fig.~\ref{fig:clusp_non_apx}.
	
	\begin{figure*}[ht]
		\centering
		\includegraphics[scale=1]{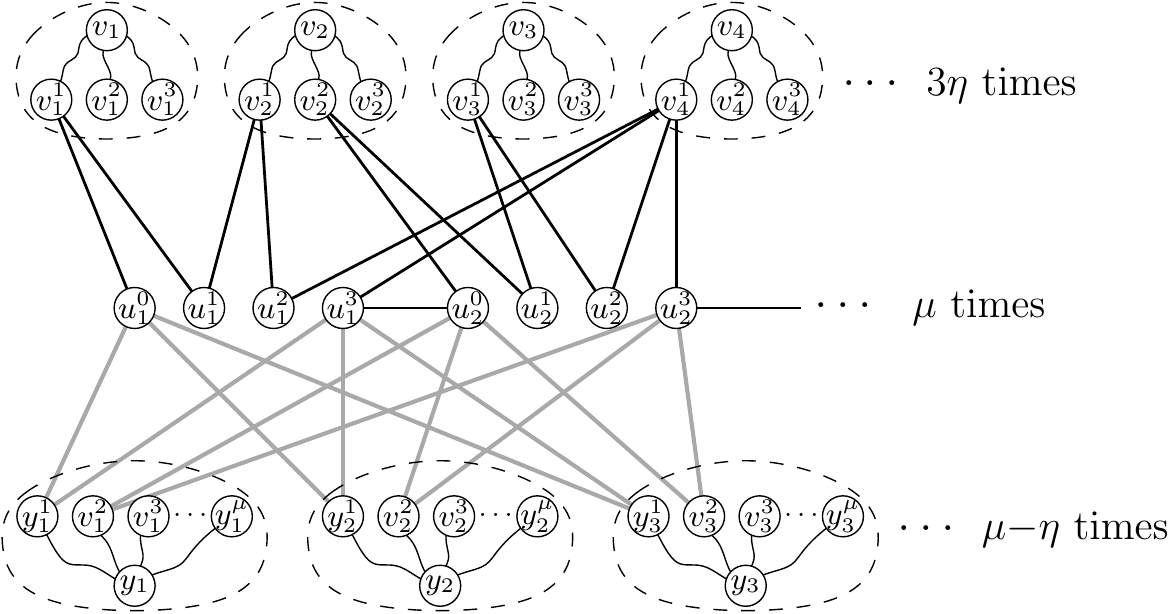}
		\caption{The graph used in the proof of Theorem~\ref{thm:clusp_non_apx}.
			Top-paths are shown with bold black edges. Bottom paths are shown with bold-gray
			edges. Paths of length $M$ are shown with curvy lines. Clusters are shown with
			dashed lines. In the corresponding X3C instance we have $S_1=\{x_1, x_2, x_4\}$
			and $S_2 = \{x_2, x_3, x_4\}$.}
		\label{fig:clusp_non_apx}
	\end{figure*}
	Now, let \opt be the cost of an optimal solution to this \clusp instance. We
	now claim that (i) if there is a solution for the X3C instance then $\opt \le 15\mu$,
	and (ii) if there is no solution to the X3C instance then $\opt \ge M$. To prove
	(i), let $\S^*$ be a solution to the X3C instance. Notice that $|\S^*|=\eta$ and
	that $|\S \setminus \S^*|=\mu-\eta$. We construct a clustered $s$-$t$-path
	$P$ as follows: for each $S_j \in \S$, if $S_j \in \S^*$ we add to $P$ all the
	edges in the top-path for $S_j$, while if $S_j \not\in \S^*$ we let $z=| \{S_1,
	\dots, S_j\} \setminus \S^*|$ and we add to $P$ all the edges in the $z$-th
	bottom-path for $S_j$. Finally, we add to $P$ all the edges in $\{ (u_j^3,
	u_{j+1}^0) \, : \, j=1, \dots, \mu-1 \}$. It is easy to see that $P$ is indeed an
	$s$-$t$-path and that each cluster is traversed only once. Moreover, $P$
	contains exactly $\eta$ top-paths (of $6$ edges each) and $\mu-\eta$ bottom
	paths (of $2$ edges each). Therefore, it follows that:
	$\opt \le 6\eta + 2(\mu-\eta) +  \mu-1 \le 4\eta + 3\mu \le 15\mu.$
	
	To prove (ii) we consider the contrapositive statement, i.e., we show that if
	$\opt < M$ then there exists a solution to the X3C instance. Let $P^*$ be an
	optimal solution to the \clusp instance and suppose $\opt < M$. This immediately
	implies that $P$ does not contain any of the paths from $y_z$ to $y_z^j$
	or any of those from $v_i$ to $v_i^h$,
	since all these paths have length $M$.
	This means that, for
	each $S_j \in \S$, $P$ contains either the (unique) top-path for $S_j$ or one of
	the bottom paths for $S_j$. Since $P$ can contain at most $\mu-\eta$
	bottom-paths and at most $\eta$ top paths (as otherwise it would violate the
	clustering constraints), it follows that $P$ contains \emph{exactly} $\eta$ top
	paths. We define $\S^*$ as the collection of the sets $S_j$ for which a top-path
	has been selected. Since $|\S^*|=\mu$ and two paths corresponding to two
	different sets in $\S^*$ cannot both pass through vertices belonging to the same
	cluster, it follows that $\S^*$ is indeed a solution to the X3C instance.
	
	We are now ready to prove the claim. Notice that the number of vertices
	of the \clusp instance, say $n$, is upper bounded by $O(\mu^2 M)$. We set
	$M=\Theta(\mu^{\frac{3}{\epsilon}-1})$ so that
	$n=O(\mu^{\frac{3}{\epsilon}+1})$. Suppose now that there exists a
	polynomial-time {$n^{1-\epsilon}$--approximation} algorithm $A$ for \clusp.
	This would imply that if the X3C instance admits a solution, then the cost of
	the solution returned by $A$ would be at most:
	\begin{align*}
	15 \mu n^{1-\epsilon} & = O(\mu
	\mu^{\frac{3}{\epsilon}-2-\epsilon}) = O(\mu^{\frac{3}{\epsilon}-1-\epsilon} )
	= O(M \mu^{-\epsilon})
	= O(M M^{-\frac{\epsilon^2}{3-\epsilon}}) = o(M)
	\end{align*}
	while, if the X3C instance does not admit a
	solution, then $A$ would return a solution to the \clusp instance having a cost
	of at least $M$. It follows that we would be able to solve X3C in polynomial
	time. \qed
\end{proof}
\section{Conclusion and Future work} \label{sec:concl}
In this paper, motivated by key modern networked applications, we have studied several clustered variants of shortest-path related problems, namely \clubfs, \cluspt and unweighted \clusp.  We have provided a comprehensive set of results which allow to shed light on the  complexity of such problems.

There are several directions that may be pursued for future work.
The main research question that we leave open is that of establishing a lower bound on the approximability of \clubfs, and, in case of a gap w.r.t. the approximation factor  provided by Algorithm~\ref{alg:clustered_bfs_apx}, that of devising a better approximation algorithm (by, e.g., exploring some other natural heuristic).
Besides that, also studying clustered shortest-path problems on restricted but meaningful classes of graphs, like, e.g., euclidean or planar graphs, might deserve investigation.
Another interesting issue is surely that of studying how other practically relevant network structures, such as spanners and highly-connected spanning subgraphs, behave in a clustered setting (incidentally, clusterization is one of the most used techniques to build this kind of structures, see e.g. \cite{BiloGG0P15}).
Finally, it would be also interesting to conduct an experimental study for assessing the practical performance of all proposed algorithms.


\bibliographystyle{abbrvurl}
\bibliography{biblio}


\end{document}